\documentclass[a4paper]{lipics-v2019}

\usepackage{stmaryrd}

\def\true#1{\llbracket #1 \rrbracket}
\def\vec#1{{\mathbf #1}}
\def\lev#1{\mathrm{lev}(#1)}
\newcommand{\inp}[2]{(\vec{#1}, \vec{#2})}
\newcommand{\conj}{\mathrm{CD}}
\newcommand{\disj}{\mathrm{DISJ}}

\newcommand{\sign}{\mathrm{sign}}

\newcommand{\gtop}{g^{\mathrm{clf}} }
\newcommand{\dom}{\{0, 1\}^{2n}}
\newcommand{\af}{\varphi}
\newcommand{\dis}{\delta}
\newcommand{\npr}{c}
\newcommand{\eq}{\mathrm{EQ}}

\newtheorem{fact}{Fact}

\title{Exponential Lower Bounds for Threshold Circuits of Sub-Linear Depth and Energy} 

\author{Kei Uchizawa}{Graduate School of Science and Engineering, Yamagata University, Jonan 4-3-16, Yonezawa-shi Yamagata, 992-8510 Japan.}{uchizawa@yz.yamagata-u.ac.jp}{}{}

\author{Haruki Abe}{Graduate School of Science and Engineering, Yamagata University, Jonan 4-3-16, Yonezawa-shi Yamagata, 992-8510 Japan.}{}{}{}

\authorrunning{K. Uchizawa and H. Abe} 

\Copyright{Kei Uchizawa and Haruki Abe} 



\nolinenumbers 

\EventEditors{}
\EventNoEds{1}
\EventLongTitle{}
\EventShortTitle{}
\EventAcronym{}
\EventYear{2021}
\EventDate{} 
\EventLocation{}
\EventLogo{}
\SeriesVolume{}
\ArticleNo{1}

\begin{document}

\maketitle

\begin{abstract}
	In this paper, we investigate computational power of threshold circuits and other theoretical models of neural networks in terms of the following four complexity measures: size (the number of gates), depth, weight and energy. Here the energy complexity of a circuit  measures sparsity of their computation, and  is defined as the maximum number of gates outputting non-zero values taken over all the input assignments.
	
	As our main result, we prove that any threshold circuit $C$ of size $s$, depth $d$, energy $e$ and weight $w$ satisfies $\log (rk(M_C)) \le ed (\log s + \log w + \log n)$, where $rk(M_C)$ is the rank of the communication matrix $M_C$ of a $2n$-variable Boolean function that $C$ computes. Thus, such a threshold circuit $C$ is able to compute only a Boolean function of which communication matrix has rank bounded by a product of logarithmic factors of $s,w$ and linear factors of $d,e$. This implies an exponential lower bound on the size of even sublinear-depth threshold circuit if energy and weight are sufficiently small. For example, we can obtain an exponential lower bound $s = 2^{\Omega(n^{1/3})}$ even for threshold circuits of depth $n^{1/3}$, energy $n^{1/3}$ and weight $2^{o(n^{1/3})}$. We also show that the inequality is tight up to a constant factor when the depth $d$ and energy $e$ satisfies $ed = o(n/ \log n)$.
	
	For other models of neural networks such as a discretized ReLE circuits and decretized sigmoid circuits,
	we prove that a similar inequality also holds for a discretized circuit $C$: $rk(M_C) = O(ed(\log s + \log w + \log n)^3)$. We obtain this inequality by showing that any discretized circuit can be simulated by a threshold circuit with moderate increase of size, depth, energy and weight. Thus, if we consider the number of non-zero output values as a measure for sparse activity of a neural network, our results suggest that larger depth linearly helps neural networks to acquire sparse activity.
\end{abstract}

\section{Introduction}

\noindent
{\bf Background}.
DiCarlo and Cox argued that constructing good internal representations is crucial to perform visual information processing, such as object recognition, for neural networks in the brain~\cite{DC07}. Here, an internal representation is described by a vector in a very high dimensional space, where each axis is one neuron's activity and the dimensionality equals to the number (e.g., $\sim$1 million) of neurons in a feedforward neural network. They call representations good if a given pair of two images that are hard to distinguish at the input space, but the resulting representations for them are easy to separate by simple classifiers such as a linear classifier. While such internal representations are likely to play fundamental role in information processing in the brain, it is also known that a neuron needs relatively high energy to be active~\cite{Lennie:03,Olshausen:04}, and hence neural networks are forced to acquire representations supported by only a small number of active neurons~\cite{Foeldiak:03}. These observations pose a question: for what information processing can neural networks construct good internal representations?

In the paper~\cite{UDM:08}, Uchizawa \emph{et al.} address the question from the viewpoint of circuit complexity. More formally, they employed threshold circuits as a model of neural networks~\cite{MP43,Minsky:88,Parberry:94,R58,Siu:91,SiuETAL:95,Siu:94}, and introduced a complexity measure, called energy complexity, for sparsity of their internal representations. A threshold circuit is a feedforward logic circuit whose basic computational element computes a linear threshold function, and energy of a circuit is defined as the maximum number of internal gates outputting ones over all the input assignments. (See also ~\cite{DOS:20,K:92,SS:22,Sun19,V:61} for studies on energy complexity of other types of logic circuits).
Uchizawa \emph{et al.} then show that the energy complexity is closely related to the rank of linear decision trees. In particular, they prove that any linear decision tree of $l$ leaves can be simulated by a threshold circuit of size $O(l)$ and energy $O(\log l)$. Thus, even logarithmic-energy threshold circuits have certain computational power: any linear decision tree of polynomial number of leaves can be simulated by a polynomial-size and logarithmic-energy threshold circuit. 

Following the paper~\cite{UDM:08}, a sequence of papers show relations among other major complexity measures such as size (the number of gates), depth, weight and fan-in~\cite{MOSUZ18,Suzuki11,Suzuki13a,uchi:08,U20,Uchizawa14,Uchizawa:10:SET}. In particular, Uchizawa and Takimoto~\cite{uchi:08} showed that any threshold circuit $C$ of depth $d$ and energy $e$ requires size $s = 2^{\Omega(n/e^d)}$ if $C$ computes a high bounded-error communication complexity function such as Inner-Product function. Even for low communication complexity functions, an exponential lower bound on the size is known for constant-depth threshold circuits: any threshold circuit $C$ of depth $d$ and energy $e$ requires size $s = 2^{\Omega(n/e2^{e+d}\log^e n)}$ if $C$ computes the parity function~\cite{U20}. These results provide exponential lower bounds if the depth is constant and energy is sub-linear~\cite{uchi:08} or sub-logarithmic~\cite{U20}, while both Inner-Product function and Parity function are computable by linear-size, constant-depth, and linear energy threshold circuits. Thus these results imply that the energy complexity strongly related to representational power of threshold circuits. However these lower bounds break down when we consider threshold circuits of larger depth and energy, say, non-constant depth and sub-linear energy.

\smallskip
\noindent
{\bf Our Results for Threshold Circuits}.
In this paper, we prove that simple Boolean functions are hard even for sub-linear depth and sub-linear energy threshold circuits. Let $C$ be a threshold circuit with Boolean input variables $\vec{x} = (x_1, x_2, \dots , x_n)$ and $\vec{y} = (y_1, y_2, \dots , y_n)$. A communication matrix $M_C$ of $C$ is a $2^n \times 2^n$  matrix where each row (resp., each column) is indexed by an assignment $\vec{a} \in \{ 0, 1\}^n$ to $\vec{x}$ (resp., $\vec{b} \in \{ 0, 1\}^n$ to $\vec{y}$), and the value $M_C[\vec{a}, \vec{b}]$ is defined to be the output of $C$ given $\vec{a}$ and $\vec{b}$. We denote by $rk(M_C)$ the rank of $M_C$ over $\mathbb{F}_2$. Our main result is the following relation among size, depth energy and weight.
\begin{theorem}\label{thm:Intro}
Let $s, d, e$ and $w$ be integers satisfying $2 \le s, d$, $10 \le e$, $1 \le w$. If a threshold circuit $C$ computes a Boolean function of $2n$ variables, and has size $s$, depth $d$, energy $e$ and weight $w$, then it holds that
\begin{equation}\label{eq:Intro:LB}
	\log ( rk(M_C) ) \le ed(\log s + \log w + \log n).
\end{equation}
\end{theorem}
The theorem implies exponential lower bounds for sub-linear depth and sub-linear energy threshold circuits. As an example, let us consider a Boolean function $\conj_n$ defined as follows: For a $2n$ input variables $x_1, \dots , x_n$ and $y_1, \dots , y_n$,
\[
	\conj_n \inp{x}{y}	= \bigwedge_{i=1}^n x_i \vee y_i.
\]
We note that $\conj_n$ is a biologically motivated Boolean function: Maass~\cite{MAASS19971659} defined $\conj_n$ to model coincidence detection or a pattern matching, and Lynch and Musco~\cite{Lynch2022} introduced a related problem, called Filter problem, for studying theoretical aspect of spiking neural networks. Since $\conj_n$ is the complement of the well-studied Boolean function, the disjointness function, the rank of $\conj_n$ is $2^n$~\cite{Jukna:11}. Thus, the theorem implies that
\begin{equation}\label{eq:Intro:LB_n}
	n \le ed(\log s + \log w + \log n)
\end{equation}
holds if a threshold circuit $C$ computes $\conj_n$. Arranging Eq.~(\ref{eq:Intro:LB_n}), we can obtain a lower bound
$2^{n/(ed)}/(wn) \le s$ which is exponential in $n$ if both $d$ and $e$ are sub-linear and $w$ is sub-exponential. For example, we can obtain an exponential lower bound $s = 2^{\Omega(n^{1/3})}$ even for threshold circuits of depth $n^{1/3}$, energy $n^{1/3}$ and weight $2^{o(n^{1/3})}$. We can obtain similar lower bounds for the Inner-Product function and the equality function, since they have linear rank.

Comparing the lower bound $s = 2^{\Omega(n/e^d)}$ given in~\cite{uchi:08} to ours, our lower bound is meaningful only for sub-exponential weight, but improves on it in two-fold: the lower bound is exponential even if $d$ is sub-linear, and provide a nontrivial lower bound for Boolean functions with much weaker condition: Threshold circuits need exponential size even for Boolean functions of the standard rank $\Omega(n)$.

Threshold circuits have received considerable attention in circuit complexity, and a number of
lower bound arguments have developed for threshold circuits under some restrictions on computational resources including size, depth, energy and weight~\cite{Amano:20,AM:05,CSS:19,Hajnal:93,Hastad:91,IPS:97,KM:16,MOSUZ18, Noam:93,RS:10,U20,uchi:08,Uchizawa:10:SET}. However, the arguments for lower bounds are designated for constant-depth threshold circuits, and hence cannot provide meaningful ones when the depth is not constant. In particular, $\conj_n$ is computable by a depth-2 and linear-size threshold circuit. Thus, directly applying known techniques are unlikely to yield an exponential lower bound for $\conj_n$.

To complement Theorem~\ref{thm:Intro}, we also show that the lower bound is tight up to a constant factor if the product of $e$ and $d$ are small:
\begin{theorem}\label{thm:Intro_UP}
	For any integers $e$ and $d$ such that $2 \le e$ and $2 \le d$, $\disj_n$ is computable by a threshold circuit of size
	\[
	s \le (e-1)(d-1)\cdot 2^{\frac{n}{(e-1)(d-1)}}.
	\]
	depth $d$, energy $e$ and weight
	\[
	w \le \left( \frac{n}{(e-1)(d-1)} \right)^2.
	\]
\end{theorem}
Substituting $s, d, e$ and $w$ of a threshold circuit given in Theorem~\ref{thm:Intro_UP} to the right hand side of Eq.~(\ref{eq:Intro:LB_n}), we have
\begin{eqnarray*}
	&&ed(\log s + \log w + \log n)\\
	&\le& ed \left( \frac{n}{(e-1)(d-1)} \! + \! \log(e-1)(d-1) \! + \! \log \left( \frac{n}{(e-1)(d-1)} \right)^2 \! + \! \log n \right)\\
	&\le& 4n + O(ed\log n),
\end{eqnarray*}
which almost matches the left hand side of Eq.~(\ref{eq:Intro:LB_n}) if $ed = o(n/\log n)$. Thus, Theorem~\ref{thm:Intro} neatly captures the computational aspect of threshold circuits computing $\conj_n$. Recall that any linear decision tree of polynomial number of leaves can be simulated by a polynomial-size and logarithmic-energy threshold circuit~\cite{UDM:08}. Also, it is known that any Boolean function is computable by a threshold circuit of energy one if exponential size is allowed~\cite{MOSUZ18}. Thus, we believe that the situation $ed = o(n/\log n)$ is not too restrictive. we We also show that the lower bound is also tight for the equality function.


\smallskip
\noindent
{\bf Our Result for Discretized Circuits}.
Besides threshold circuits, we consider other other well-studied model of neural network, where an activation function and weights of an computational element are discretized (such as, discretized sigmoid or ReLU circuits). The size, depth, energy and weight are important parameters also for artificial neural networks. The size and depth are major topics on success of deep learning. The energy is related to important techniques for deep learning method such as regularization, sparse coding, or sparse autoencoder~\cite{He:14,Lee:06,Ng:11}. The weight resolution is closely related to chip resources in neuromorphic hardware systems~\cite{PPSPSDM12}, and quantization schemes received attention~\cite{CourbariauxEtAl15,HubaraEtAl:18}. 

We define similar notions for the energy and weight of a discretized circuit, and show that any discretized circuit can be simulated by a threshold circuit with a moderate increase in size, depth, energy, and weight. Consequently, combining with Theorem~\ref{thm:Intro}, we can show that the rank is bounded by a product of the polylogarithmic factors of $s,w$ and linear factors of $d, e$ for discretized circuits. For example, we can obtain the following proposition for sigmoid circuits: 
\begin{theorem}\label{thm:Intro_Dscr}
	If a sigmoid circuit $C$ of size $s$, depth $d$, energy $e$, and weight $w$ computes a Boolean function $f$, then it holds that
	\begin{equation*}
		\log ( rk(M_C) ) =O (ed(\log s+  \log w + \log n)^3).
	\end{equation*}
\end{theorem}
Maass, Schnitger and Sontag~\cite{MSS91} showed that a sigmoid circuit could be simulated by a threshold circuit, but their simulation was optimized to be depth-efficient and did not consider energy. Thus, their result does not fit into our purpose. 

Theorems~\ref{thm:Intro} and~\ref{thm:Intro_Dscr} imply that a threshold circuit or discretized circuit are able to compute a Boolean function of bounded rank. Thus, we can consider these theorems as bounds on corresponding concept classes. According to the bound, $c$ times larger depth is comparable to $2^c$ times larger size. Thus, large depth could enormously help neural networks to increase its expressive power. Also, the bound suggests that increasing depth could also help a neural network to acquire sparse activity when we have hardware constraints on both the number of neurons and the weight resolution. These observations may shed some light on the reason for the success of deep learning.

\smallskip
\noindent
{\bf Organization}.
The rest of the paper is organized as follows. In Section 2, we define terms needed for analysis. In Section 3, we present our main lower bound result. In Section 4, we show the tightness of the bound by constructing a threshold circuit computing $\conj_n$. In Section 5, we show that a discretized circuit can be simulated by a threshold circuit. In Section 6, we conclude with some remarks.

\section{Preliminaries}

For an integer $n$, we denote by $[n]$ a set $\{1, 2, \dots n\}$. The base of the logarithm is two unless stated otherwise. In Section~\ref{sec:circuit}, we define terms on threshold circuits and discretized circuits. In Section~\ref{sec:matrx}, we define communication matrix, and present some known facts. 

\subsection{Circuit Model}\label{sec:circuit}
In Sections~\ref{sec:TC} and~\ref{sec:DC}, we give definitions of threshold and discritized circuits, respectively.

\subsubsection{Threshold Circuits}\label{sec:TC}
Let $k$ be a positive integer. A \emph{threshold gate} $g$ with $k$ input variables $\xi_1, \xi_2, \dots, \xi_k$ has weights $w_1, w_2, \dots , w_k$, and a threshold $t$. We define the output $g(\xi_1, \xi_2, \dots, \xi_k)$ of $g$ as
\begin{eqnarray*}
	g (\xi_1, \xi_2, \dots, \xi_k ) = \sign\left(\sum_{i=1}^k w_i \xi_i - t \right)
	=
	\left\{
	\begin{array}{ll}
		1 & \mbox{ if } t \le \sum_{i=1}^k w_i \xi_i;\\
		0 & \mbox{ otherwise}
	\end{array}
	\right.
\end{eqnarray*}
 To evaluate the weight resolution, we assume single synaptic weight to be discrete, and that $w_1, w_2, \dots , w_n$ are integers. The \emph{weight $w_g$ of $g$} is defined as the maximum of the absolute values of $w_1, w_2, \dots , w_k$. In other words, we assume that $w_1, w_2, \dots , w_k$ are $O(\log w_g)$-bit coded discrete values. Throughout the paper, we allow a gate to have both positive and negative weights, although biological neurons are either excitatory (all the weights are positive) or inhibitory (all the weights are negative). As mentioned in~\cite{MAASS19971659}, this relaxation has basically no impact on circuit complexity investigations, unless one cares about constant blowup in computational resources. This is because a single gate with positive and negative weights can be simulated by a pair of excitatory and inhibitory gates.

A \emph{threshold circuit} $C$ is a combinatorial circuit consisting of threshold gates, and is expressed by a directed acyclic graph. The nodes of in-degree 0 correspond to input variables, and the other nodes correspond to gates.
Let $G$ be a set of the gates in $C$. For each gate $g \in G$, the \emph{level} of $g$, denoted by $\lev{g}$, is defined as the length of a longest path from an input variable to $g$ on the underlying graph of $C$. For each $l \in [d]$, we define $G_l$ as a set of gates in the $l$th level: 
\[
	G_l = \{ g \in G \mid \lev{g} = l\}.
\]
Given an input assignment $(\vec{a}, \vec{b}) \in \dom$ to $(\vec{x}, \vec{y})$, the outputs of the gates in $C$ are inductively determined from the bottom level.

In this paper, we consider a threshold circuit $C$ for a Boolean function $f: \dom \to \{ 0, 1\}$. Thus, $C$ has $2n$ Boolean input variables $\vec{x} = (x_1, x_2, \dots , x_n)$ and $\vec{y} = ( y_1, y_2, \dots , y_n)$, and a unique output gate, denoted by $\gtop$, which is a linear classifier separating internal representations given by the gates in the lower levels (possibly together with input variables).
Consider a gate $g$ in $C$. Let $w^x_1, w^x_2, \dots , w^x_n$ (resp., $w^y_1, w^y_2, \dots , w^y_n$) be the weights for $x_1, x_2, \dots, x_n$ (resp., $y_1, y_2, \dots, y_n$), and $t_g$ be threshold of $g$. For each gate $h$ directed to $g$, let $w_{h, g}$ be a weight of $g$ for the output of $h$. Then the output $g(\vec{x}, \vec{y})$ of $g$ is defined as
\[
	g(\vec{x}, \vec{y}) =
		\sign\left( p_g\inp{x}{y} - t_g\right)
\]
where $p_g\inp{x}{y}$ denotes a potentials of $g$ invoked by the input variables and gates:
\[
	p(\vec{x}, \vec{y}) = \sum_{i=1}^n w^x_ix_i +  \sum_{i=1}^n w^y_i y_i + \sum_{l = 1}^{\mbox{{\scriptsize lev}}(g)-1} \sum_{h \in G_{l}} w_{h, g}h(\vec{x}, \vec{y}).
\]
We sometimes write $p^x_g(\vec{x})$ (resp., $p^y_g(\vec{y})$) for the potential invoked by $\vec{x}$ (resp., $\vec{y}$):
\[
	p^x_g(\vec{x}) = \sum_{i=1}^n w^x_ix_i \quad \mbox{ and } \quad p^y_g(\vec{y}) = \sum_{i=1}^n w^y_iy_i.
\]
Although the inputs to $g$ are not only $\vec{x}$ and $\vec{y}$ but the outputs of gates in the lower levels, we write $g(\vec{x}, \vec{y})$ for the output of $g$, because $\vec{x}$ and $\vec{y}$ inductively decide the output of $g$. We say that $C$ \emph{computes} a Boolean function $f: \dom \to\{0, 1\}$ if $\gtop (\vec{a}, \vec{b})=f(\vec{a}, \vec{b})$ for every $\inp{a}{b}\in \dom$.

Let $C$ be a threshold circuit. We define \emph{size} $s$ of $C$ as the number of the gates in $C$, and \emph{depth} $d$ of $C$ as the level of $\gtop$. We define the \emph{energy} $e$ of $C$ as
\[
	e=\max _{(\vec{a}, \vec{b}) \in\dom} \sum_{g\in G} g(\vec{a}, \vec{b}).
\]
We define \emph{weight} $w$ of $C$ as the maximum of the weights of the gates in $C$: $w = \max_{g \in G} w_g$.

\subsubsection{Discretized Circuits}\label{sec:DC}

Let $\af$ be an activation function. Let $\dis$ be a discretizer that maps a real number to a number representable by a bitwidth $b$. We define a discretized activation function $\dis \circ \af$ as a composition of $\af$ and $\delta$, that is, $\dis \circ \af(x) = \dis(\af(x))$ for any number $x$. We say that $\dis \circ \af$ has \emph{silent range for an interval $I$} if $\delta \circ \af (x) = 0$ if $x \in I$, and $\delta \circ \af (x) \neq 0$, otherwise. For example, if we use the ReLU function as the activation function $\af$, then $\delta \circ \af$ has silent range for $I = (-\infty, 0]$ for any discretizer $\dis$. If we use the sigmoid function as the activation function $\af$ and linear partition as discretizer $\delta$, then $\delta \circ \af$ has silent range for $I = (-\infty, t_{\max}]$ where $t_{\max} = \ln (1/(2^b-1))$ where $\ln$ is the natural logarithm.

Let $\delta \circ \af$ be a discretized activation function with silent range. A \emph{$(\delta \circ \af)$-gate} $g$ with $k$ input variables $\xi_1, \xi_2, \dots, \xi_k$ has weights $w_1, w_2, \dots , w_k$ and a threshold $t$, where each of the weights and threshold are discretized by $\delta$. The output $g(\xi_1, \xi_2, \dots, \xi_k)$ of $g$ is then defined as
\begin{eqnarray*}
	g (\xi_1, \xi_2, \dots, \xi_k ) =  \delta \circ \af \left(\sum_{i=1}^k w_i \xi_i - t \right).
\end{eqnarray*}
A \emph{$(\delta \circ \af)$-circuit} is a combinatorial circuit consisting of $(\delta \circ \af)$-gates except that the top gate $\gtop$ is a threshold gate, that is, a linear classifier. We define size and depth of a $(\delta \circ \af)$-circuit same as the ones for a threshold circuit. We define energy $e$ of a $(\delta \circ \af)$-circuit as the maximum number of gates outputting non-zero values in the circuit:
\[
	e=\max _{(\vec{a}, \vec{b}) \in \dom} \sum_{g\in G} \true{g(\vec{a}, \vec{b}) \neq 0}
\]
where $\true{\mathrm{P}}$ for a statement $P$ denote a notation of the function which outputs one if $\mathrm{P}$ is true, and zero otherwise. We define \emph{weight} $w$ of $C$ as $w = 2^{2b}$, where $2b$ is the bitwidth possibly needed to represent a potential value invoked by a single input of a gate in $C$.

\subsection{Communication Matrix and its Rank}\label{sec:matrx}
Let $Z \subseteq \{ 0, 1\}^n$. For a Boolean function $f:Z \times Z \to \{ 0, 1\}$, we define a communication matrix $M_f$ over $Z$  as a $2^{|Z|} \times 2^{|Z|}$ matrix where each row and column are indexed by $\vec{a} \in Z$ and $\vec{b} \in Z$, respectively, and each entry is defined as $M_f(\vec{a}, \vec{b}) = f(\vec{a}, \vec{b})$. We denote by $rk(M_f)$ the rank of $M_f$ over $\mathbb{F}_2$. If a circuit $C$ computes $f$, we may write $M_C$ instead of $M_f$. If a Boolean function $f$ does not have an obvious separation of the input variables to $\vec{x}$ and $\vec{y}$, we may assume a separation so that $rk(M_f)$ is maximized.

Let $k$ and $n$ be natural numbers such that $k \le n$. Let
\[
	Z_k = \{ \vec{a} \in \{ 0, 1\}^n \mid \mbox{The number of ones in $\vec{a}$ is at most $k$} \}.
\]
A $k$-disjointness function $\mathrm{DISJ}_{n, k}$ over $Z_k$ is defines as follows:
\[
	\mathrm{DISJ}_{n, k} \inp{x}{y}	= \bigwedge_{i=1}^n \overline{x_i} \vee \overline{y_i}
\]
where the input assignments are chosen from $Z_k$. The book~\cite{Jukna:11} contains a simple proof showing $\mathrm{DISJ}_{n, k}$ has full rank. 
\begin{theorem}[Theorem 13.10 \cite{Jukna:11}]\label{thm:rank_of_disj}
	$rk(M_{\mathrm{DISJ}_{n, k}}) = \sum_{i=0}^k { n \choose i}$.
	In particular, $rk(M_{\mathrm{DISJ}_{n, n}}) = 2^n$.
\end{theorem}
$\conj_n$ is the complement of $\mathrm{DISJ}_{n, n}$. We can obtain the same bound for $\conj_n$, as follows:
\begin{corollary}\label{thm:rank_of_conj}
	$rk(M_{\conj_n}) = 2^n$.
\end{corollary}

We also use well-known facts on the rank. Let $A$ and $B$ be two matrices of same dimensions. We denote by $A+B$ the summation of $A$ and $B$, and by $A\circ B$ the Hadamard product of $A$ and $B$.
\begin{fact}
	For two matrices $A$ and $B$ of same dimensions, we have
	\begin{description}
		\item[(i)] $rk(A + B) \le rk(A) + rk(B)$;
		\item[(ii)] $rk(A \circ B) \le rk(A) \cdot rk(B)$.
	\end{description}
\end{fact}

\section{Lower Bound for Threshold Circuits}

In this section, we give the inequality relating the rank of the communication matrix to the size, depth, energy and weight.
\begin{theorem}[Theorem 1 restated]\label{thm:LB_size}
	Let $s, d, e$ and $w$ be integers satisfying $2 \le s, d$, $11\le e$, $1 \le w$. Suppose a threshold circuit $C$ computes a Boolean function of $2n$ variables, and has size $s$, depth $d$, energy $e$, and weight $w$. Then it holds that
	\[
	\log ( rk(M_C) ) \le ed(\log s + \log w + \log n).
	\]
\end{theorem}
We prove the the theorem by showing that $M_C$ is a sum of matrices each of which corresponds to an internal representation that arises in $C$. Since $C$ has bounded energy, the number of internal representations is also bounded. We then show by the inclusion-exclusion principle that each matrix corresponding an internal representation has bounded rank. Thus, Fact 1 implies the theorem.
\begin{proof}
	Let $C$ be a threshold circuit that computes a Boolean function of $2n$ variables, and has size $s$, depth $d$, energy $e$ and weight $w$. Let $G$ be a set of the gates in $C$. For $l \in [d]$, let $G_l$ be a set of the gates in $l$-th level of $C$. Without loss of generality, we assume that $G_d = \{ \gtop \}$. We evaluate the rank of $M_C$, and prove that
	\begin{eqnarray}\label{eq:UB_rk(M_C)}
		rk(M_C) \le \left(\frac{\npr\cdot s}{e-1} \right)^{e-1} \! \cdot \left(\left(\frac{\npr\cdot s}{e-1} \right)^{e-1} \! \cdot (2nw+1)^{e-1} \right)^{d-1} \! \cdot (2nw+1)	
	\end{eqnarray}
	where $\npr < 3$. Equation~(\ref{eq:UB_rk(M_C)}) implies that
	\begin{eqnarray*}
		rk(M_C)
		&\le& \left(\frac{\npr\cdot s}{e-1}\cdot (2nw+1) \right)^{(e-1)d}\\
		&\le& \left(snw \right)^{ed},
	\end{eqnarray*}
	where the last inequality holds if $e \ge 11$. Taking the logarithm of the inequality, we obtain the theorem.
	
	Below we verify that Eq.~(\ref{eq:UB_rk(M_C)}) holds. Let $\vec{P} = (P_1, P_2, \dots , P_d)$, where $P_l$ is defined as a subset of $G_l$ for each $l \in [d]$. Given an input $(\vec{a}, \vec{b}) \in \dom$, we say that an \emph{internal representation $\vec{P}$ arises for $(\vec{a}, \vec{b})$} if, for every $l \in [d]$,
	$g (\vec{a}, \vec{b}) = 1$ for every $g \in P_l$, and $g (\vec{a}, \vec{b}) = 0$ for every $g \not\in P_l$. We denote by $\vec{P}^*(\vec{a}, \vec{b})$ the internal representation that arises for $(\vec{a}, \vec{b}) \in \dom$. We then define $\mathcal{P}_1$ as a set of internal representations that arise for $\inp{a}{b}$ such that $\gtop \inp{a}{b} =   1$:
	\[
	\mathcal{P}_1 = \{\vec{P}^*(\vec{a}, \vec{b}) \mid \gtop \inp{a}{b} =   1\}.
	\]
	Note that, for any $\vec{P} = (P_1, P_2, \dots , P_d) \in \mathcal{P}_1(C)$, 
	\[
	|P_1|+ |P_2| + \cdots + |P_{d-1}| \le e-1
	\]
	and $|P_d| = 1$. Thus we have
	\begin{eqnarray}\label{eq:UB_P_1}
		|\mathcal{P}_1| \le \sum_{k=0}^{e-1} {s \choose k} \le \left(\frac{\npr\cdot s}{e-1} \right)^{e-1}.
	\end{eqnarray}
	
	For each $\vec{P} \in \mathcal{P}_1$, let $M_{\vec{P}}$ be a $2^n \times 2^n$ matrix such that, for every $(\vec{a}, \vec{b}) \in \dom$,
	\[
	M_{\vec{P}}(\vec{a}, \vec{b}) = \left\{
	\begin{array}{ll}
		1 & \mbox{ if } \vec{P} = \vec{P}^*(\vec{a}, \vec{b});\\
		0 & \mbox{ if } \vec{P} \neq \vec{P}^*(\vec{a}, \vec{b}).
	\end{array}
	\right.
	\]
	By the definitions of $\mathcal{P}_1$ and $M_{\vec{P}}$, we have
	\[
	M_C = \sum_{\vec{P}\in \mathcal{P}_1} M_{\vec{P}},
	\]
	and hence Fact 1(i) implies that
	\[
	rk(M_C) \le \sum_{\vec{P}\in \mathcal{P}_1} rk(M_{\vec{P}} ).
	\]
	Thus Eq.~(\ref{eq:UB_P_1}) implies that
	\[
	rk(M_{\vec{P}}) \le \left(\frac{\npr\cdot s}{e-1} \right)^{e-1}
	\cdot \max_{\vec{P}\in \mathcal{P}_1} rk(M_{\vec{P}}).
	\]
	We complete the proof by showing that, for any $\vec{P} \in \mathcal{P}_1(C)$, it holds that
	\[
	rk(M_{\vec{P}}) \le \left(\left(\frac{2\npr\cdot s}{e-1} \right)^{e-1}
	\cdot (2nw+1)^{e-1} \right)^{d-1} \cdot (2nw+1).
	\]

	In the following argument, we consider an arbitrary fixed internal representation $\vec{P} = (P_1, P_2, \dots , P_d)$ in $\mathcal{P}_1$.  We call a gate a \emph{threshold function} if the inputs of the gate consists of only $\vec{x}$ and $\vec{y}$. For each $g \in G$, we denote by $\tau[g, \vec{P}]$ a threshold function defined as
	\[
	\tau[g, \vec{P}](\vec{x}, \vec{y}) = \sign \left(p^x_g(\vec{x}) + p^y_g(\vec{y}) +  t_g[\vec{P}]\right).
	\]
	where $t_g[\vec{P}]$ is a threshold of $g$, being assumed that the internal representation $\vec{P}$ arises:
	\[
	t_g[\vec{P}] = \sum_{l = 1}^{{\lev{g}}-1} \sum_{h \in P_{l}} w_{h, g} - t_g.
	\]
	For each $l \in [d]$, we define a set $T_l$ of threshold functions as
	\[
	T_l= \{ \tau[g, \vec{P}] \mid g \in G_l\}.
	\]
	Since every gate in $G_1$ is a threshold function, $T_1$ is identical to $G_1$. 
	
	For any set $T$ of threshold functions, we denote by $M[T]$ a $2^n \times 2^n$ matrix such that, for every $(\vec{a}, \vec{b}) \in \dom$,
	\[
	M[T](\vec{a}, \vec{b}) = \left\{
	\begin{array}{ll}
		1 & \mbox{ if } \forall  \tau\in T, \tau(\vec{a}, \vec{b}) = 1;\\
		0 & \mbox{ if } \exists  \tau \in T, \tau(\vec{a}, \vec{b}) = 0.
	\end{array}
	\right.
	\]
	It is well-known that the rank of $M[T]$ is bounded~\cite{Forster:01,Hajnal:93}, as follows.
	\begin{claim}\label{clm:UB_RankofTP}
		$rk(M[T]) \le (2nw+1)^{|T|}$.
	\end{claim}
	We give a proof of the claim in Appendix for completeness.
	
	For each $l\in [d]$, based on $P_l$ in $\vec{P}$, we define a set $Q_l$ of threshold functions as
	\[
	Q_l= \{ \tau[g, \vec{P}] \mid g \in P_l\} \subseteq T_l
	\]
	and a family $\mathcal{P}(Q_l)$ of sets $T$ of threshold functions as
	\[
	\mathcal{T}(Q_l) = \{ T \subseteq T_l \mid  Q_l \subseteq T \mbox{ and } |T| \le e-1\}.
	\]
	Following the inclusion-exclusion principle, we define a $2^n \times 2^n$ matrix 
	\[
	H[Q_l] = \sum_{T \in \mathcal{T}(Q_l)} (-1)^{|T| - |Q_l|} M[T].
	\]
	We can show that $M_{\vec{P}}$ is expressed as the Hadamard product of $H[Q_1], H[Q_2], \dots , H[Q_d]$: 
	\begin{claim}\label{clm:M_P=Prod_H}
		\[
		M_{\vec{P}} = H[Q_1] \circ H[Q_2] \circ \dots \circ H[Q_d].
		\]
	\end{claim}
	The proof of the claim is given in Appendix.
	
	We finally evaluate $rk(M_{\vec{P}})$.
	Claim~\ref{clm:M_P=Prod_H} and Fact 1(ii) imply that
	\begin{equation}\label{eq:Mp_le_prod_H[t_l]}
		rk(M_\vec{P}) = rk\left( 	H[Q_1] \circ H[Q_2] \circ \dots \circ H[Q_d] \right) \le \prod_{l=1}^d rk(H[Q_l]).
	\end{equation}
	Since
	\[
	|\mathcal{T}(Q_l)| \le \left(\frac{\npr\cdot s}{e-1} \right)^{e-1}
	\]
	Fact 1(i) and Claim~\ref{clm:UB_RankofTP} imply that
	\begin{eqnarray}\label{eq:H[T^*_l]<sr^e}
		rk(H[Q_l])
		&\le& \sum_{T \in \mathcal{P}(Q_l)}  rk( M[T] ) \nonumber\\
		&\le& \left(\frac{\npr\cdot s}{e-1} \right)^{e-1}  \cdot (2nw+1)^{e-1}
	\end{eqnarray}
	for every $l \in [d-1]$, and
	\begin{eqnarray}\label{eq:H[T^*_d]<r}
		rk(H[Q_d]) \le 2nw+1.
	\end{eqnarray}
	Equations~(\ref{eq:Mp_le_prod_H[t_l]})-(\ref{eq:H[T^*_d]<r}) imply that
	\[
	rk(M_\vec{P}) \le \left(\left(\frac{\npr\cdot s}{e-1} \right)^{e-1}
	\cdot (2nw+1)^{e-1} \right)^{d-1} \cdot (2nw+1)
	\]
	as desired. We thus have verified Eq.~(\ref{eq:UB_rk(M_C)}).
\end{proof}

Combining Theorem~\ref{thm:rank_of_conj} and Corollary \ref{thm:LB_size}, we obtain the following corollary:
\begin{corollary}\label{cor:LB_conj}
	Let $s, d, e$ and $w$ be integers satisfying $2 \le s, d$, $10 \le e$, $1 \le w$. Suppose a threshold circuit $C$ of size $s$, depth $d$, energy $e$, and weight $w$ computes $\conj_n$. Then it holds that
	\[
	n \le ed(\log s + \log w + \log n).
	\]
	Equivalently, we have $2^{n/(ed)}/(nw) \le s$.
\end{corollary}

Theorem~\ref{thm:LB_size} implies lower bounds for other Boolean functions with linear rank.
For example, consider another Boolean function $\eq_n$ asking if $\vec{x} = \vec{y}$:
\[
\eq_n\inp{x}{y} = \bigwedge_{i=1}^n \overline{x_i \oplus y_i}
\]
Since $M_{\eq_n}$ is the identity matrix with full rank, we have the same lower bound.
\begin{corollary}\label{cor:LB_eq}
	Let $s, d, e$ and $w$ be integers satisfying $2 \le s, d$, $10 \le e$, $1 \le w$. Suppose a threshold circuit $C$ of size $s$, depth $d$, energy $e$, and weight $w$ computes $\eq_n$. Then it holds that
	\[
	n \le ed(\log s + \log w + \log n).
	\]
	Equivalently, we have $2^{n/(ed)}/(nw) \le s$.
\end{corollary}

\section{Tightness of the Lower Bound}

In this section, we show that the lower bound given in Theorem~\ref{thm:LB_size} is almost tight if the depth and energy are small.

\subsection{Definitions}

Let $z$ be a positive integer, and $f$ be a Boolean function of $2n$ variables. We say that $f$ is $z$-piecewise with $f_1, f_2, \dots , f_z$ if the following conditions are satisfied: Let
\[
B_j = \{ i \in [n]\mid \mbox{$x_i$ or $y_i$ are fed into $f_j$} \},
\]
then
\begin{description}
	\item[(i)] $B_1, B_2, \dots , B_z$ compose a partition of $[n]$;
	\item[(ii)] $|B_j| \le \lceil n/z \rceil$ for every $j\in [z]$;
	\item[(iii)]
	\[
	f(\vec{x}, \vec{y}) = \bigvee_{j=1}^z f_j(\vec{x}, \vec{y}) \quad \mbox{or} \quad f(\vec{x}, \vec{y}) = \overline{\bigvee_{j=1}^z f_j(\vec{x}, \vec{y})}.
	\]
\end{description}

We say that a set of threshold gates sharing input variables is a neural set, and a neural set is selective if at most one of the gates in the set outputs one for any input assignment. A selective neural set $S$ computes a Boolean function $f$ if for every assignment in $f^{-1}(0)$, no gates in $S$ outputs one, while for every assignment in $f^{-1}(1)$, exactly one gate in $S$ outputs one. We define the size and weight of $S$ as $|S|$ and $\max_{g\in S} w_g$, respectively.

By a DNF-like construction, we can obtain a selective neural set of exponential size that computes $f$ for any Boolean function $f$.
\begin{theorem}\label{thm:neuralset}
	For any Boolean function $f$ of $n$ variables, there exists a selective neural set of size $2^n$ and weight one that computes $f$.
\end{theorem}

\subsection{Upper Bounds}

The following proposition shows that we can a construct threshold circuits of small energy for piecewise functions.
\begin{lemma}\label{lem:UB}
	Let $e$ and $d$ be integers satisfying $2 \le e$ and $2 \le d$, and $z = (e-1)(d-1)$. Suppose $f:\dom \to \{ 0, 1\}$ is a $z$-piecewise function with $f_1, f_2, \dots , f_z$. If $f_j$ is computable by a selective neural set of size at most $s'$ and weight $w'$ for every $j \in [z]$, $f$ is computable by a threshold circuit of size
	\[
	s \le z\cdot s' + 1,
	\]
	depth $d$, energy $e$ and weight
	\[
	w \le  \frac{2n}{z}\cdot w'.
	\]
\end{lemma}

Clearly, $\conj_n$ is $z$-piecewise, and so the lemma gives our upper bound for $\conj_n$.
\begin{theorem}[Theorem~\ref{thm:Intro_UP} restated]\label{thm:UB_conj}
	For any integers $e$ and $d$ such that $2 \le e$ and $2 \le d$, $\conj_n$ is computable by a threshold circuit of size
	\[
	s \le (e-1)(d-1)\cdot 2^{\frac{n}{(e-1)(d-1)}}.
	\]
	depth $d$, energy $e$ and weight
	\[
	w \le \left( \frac{n}{(e-1)(d-1)} \right)^2.
	\]
\end{theorem}

We can also obtain a similar proposition for $\eq_n$.
\begin{theorem}\label{thm:UB_eq}
	For any integers $e$ and $d$ such that $2 \le e$ and $2 \le d$, $\eq_n$ is computable by a threshold circuit of size
	\[
	s \le (e-1)(d-1)\cdot 2^{\frac{2n}{(e-1)(d-1)}}.
	\]
	depth $d$, energy $e$ and weight
	\[
	w \le \frac{n}{(e-1)(d-1)}.
	\]
\end{theorem}

\section{Simulating Discretized Circuits}

In this section, we show that any discretized circuit can be simulated using a threshold circuit with a moderate increase in size, depth, energy, and weight. Thus, a similar inequality holds for discretized circuits. Recall that we define energy of discretized circuits differently from those of threshold circuits.

\begin{theorem}[Theorem~\ref{thm:Intro_Dscr} restated]\label{thm:LB_discretizer}
	Let $\dis$ be a discretizer and $\af$ be an activation function such that $\dis \circ \af$ has a silent range. If a $(\dis \circ \af)$-circuit $C$ of size $s$, depth $d$, energy $e$, and weight $w$ computes a Boolean function $f$, then it holds that
	\begin{equation*}\label{eq:Discrized_LB}
		\log ( rk(M_C) ) =O (ed(\log s+  \log w + \log n)^3).
	\end{equation*}
\end{theorem}

Our simulation is based on a binary search of the potentials of a discretized gate, we employ a conversion technique from a linear decision tree to a threshold circuit presented in~\cite{UDM:08}.

\section{Conclusions}

In this paper, we consider threshold circuits and other theoretical models of neural networks in terms of four complexity measures, size, depth energy and weight. We prove that if the communication matrix of a Boolean function $f$ has linear rank, any threshold circuit of sub-linear depth, sub-linear energy and sub-exponential weight needs exponential size to compute $f$.

We believe that circuit complexity arguments provide insights into neural computation. Besides the famous Marr's three-level approach towards understanding the brain computation (computational level, algorithmic level, and implementation level)~\cite{Mar82}, Valiant~\cite{V14} added an additional requirement that it has to incorporate some understanding of the quantitative constraints that are faced by cortex. We expect that circuit complexity arguments could shed light on a quantitative constraint through a complexity measure. Maass \emph{et al}.~\cite{MPVL19} pointed out a difficulty to uncover a neural algorithm employed by the brain, because its hardware could be extremely adopted to the task, and consequently the algorithm vanishes: even if we know to the last detail its precise structure, connectivity, and vast array of numerical parameters (an artificial neural network given by deep learning is the case), it is still hard to extract an algorithm implemented in the network. A lower bound does not provide a description of an explicit neural algorithm, but could give a hint for computational natures behind neural computation, because a lower bound argument necessarily deals with \emph{every} algorithm which a theoretical model of a neural network is able to implement. We expect that developing lower bound proof techniques for a theoretical model of neural networks can push this line of research.

\bibliography{../../uchizawa}

\newpage

\section{Appendix}

\subsection{Proof of Claim~\ref{clm:UB_RankofTP}}

	Let $z = |T|$, and $\tau_1, \tau_2, \dots , \tau_z$ be an arbitrary order of threshold functions in $T$. For each $k \in [z]$, we define
	\[	
	R_k =  \{ p^x_{\tau_k}(\vec{a}) \mid \vec{a} \in \{0, 1\}^n\}.
	\]
	Since a threshold function receives a value between $-w$ and $w$ from a single input, we have $|R_k| \le 2nw+1$.
	For $\vec{r} = (r_1, r_2 , \dots , r_z)\in R_1 \times R_2 \times \dots  \times R_z$,  we define $R(\vec{r}) = X(\vec{r}) \times Y(\vec{r})$ as a combinatorial rectangle where
	\[
	X(\vec{r}) = \{ \vec{x} \mid \forall k \in [z], p_{\tau_k}(\vec{x}) = r_k\}
	\]
	and
	\[
	Y(\vec{r}) = \{ \vec{y} \mid  \forall k \in [z], t_{\tau_k} \le r_k + p^y_{\tau_k}(\vec{y}) \}.
	\]
	Clearly, all the rectangles are disjoint, and hence $M[T]$ can be expressed as a sum of rank-1 matrices given by $R(\vec{r})$'s taken over all the $\vec{r}$'s. Thus Fact 1(i) implies that its rank is at most $|R_1 \times R_2 \times \dots  \times R_z| \le (2nw+1)^z$.

\subsection{Proof of Claim~\ref{clm:M_P=Prod_H}}
	Consider an arbitrary fixed assignment  $(\vec{a}, \vec{b}) \in \dom$. We show that
	\[
	H[Q_1](\vec{a}, \vec{b}) \circ H[Q_2](\vec{a}, \vec{b}) \circ \dots \circ H[Q_d](\vec{a}, \vec{b})= 0,
	\]
	if $M_{\vec{P}}(\vec{a}, \vec{b}) = 0$, and
	\[
	H[Q_1](\vec{a}, \vec{b}) \circ H[Q_2](\vec{a}, \vec{b}) \circ \dots \circ H[Q_d](\vec{a}, \vec{b})= 1,
	\]
	if $M_{\vec{P}}(\vec{a}, \vec{b}) = 1$. We write $\vec{P}^* = (P^*_1, P^*_2, \dots , P_d^*)$ to denote $\vec{P}^*(\vec{a}, \vec{b})$ for a simpler notation. 
	
	Suppose $M_{\vec{P}}(\vec{a}, \vec{b}) = 0$. In this case, we have $\vec{P} \neq \vec{P}^*$, and hence there exists a level $l \in [d]$ such that $P_l \neq P^*_l$ while $P_{l'} = P^*_{l'}$ for every $l' \in [l-1]$. For such $l$, it holds that
	\begin{eqnarray}\label{eq:tauP^*=tauP}
		\tau[g, \vec{P}^*](\vec{a}, \vec{b}) = \tau[g, \vec{P}](\vec{a}, \vec{b})
	\end{eqnarray}
	for every $g \in G_l$. We show that $H[Q_l](\vec{a}, \vec{b}) = 0$ by 
	considering two cases: $P_l\backslash P^*_l \neq \emptyset$ and  $P_l \subset P_l^*$.
	
	Consider the case where $P_l\backslash P^*_l \neq \emptyset$, then there exists $g \in P_l\backslash P^*_l$. Since $g \not\in P^*_l$, we have $\tau[g, \vec{P}^*](\vec{a}, \vec{b}) =0$. Thus, Eq.~(\ref{eq:tauP^*=tauP}) implies that $\tau[g, \vec{P}](\vec{a}, \vec{b}) = 0$, and hence $M[T](\vec{a}, \vec{b}) = 0$ for every $T$ such that $Q_l \subseteq T$. Therefore, for every $T \in \mathcal{T}(Q_l)$, we have $M[T](\vec{a}, \vec{b}) = 0$, and hence 
	\[
	H[Q_l](\vec{a}, \vec{b}) = \sum_{T \in \mathcal{T}(Q_l)} M[T](\vec{a}, \vec{b}) = 0.
	\]
	
	Consider the other case where $P_l \subset P_l^*$. Let $Q^*_l = \{ \tau[g, \vec{P}^*] \mid g \in P^*_l\}$. Equation~(\ref{eq:tauP^*=tauP}) implies that $M[T](\vec{a}, \vec{b}) = 1$ if $T$ satisfies $Q_l \subseteq T \subseteq Q^*_l$, and $M[T] \inp{a}{b} = 0$, otherwise. Thus, 
	\begin{eqnarray*}
		H[Q_l](\vec{a}, \vec{b})
		&=& \sum_{T \in \mathcal{T}(Q_l)} (-1)^{|T| - |Q_l|} M[T]\\
		&=& \sum_{Q_l \subseteq T \subseteq Q^*_l} (-1)^{|T| - |Q_l|}
	\end{eqnarray*}
	Therefore, by the binomial theorem,
	\begin{eqnarray*}
		H[Q_l](\vec{a}, \vec{b})
		&=& \sum_{k=0}^{|Q^*_l|-|Q_l|} {|Q^*_l|-|Q_l| \choose k} (-1)^k\\
		&=& (1 - 1)^{|Q^*_l|-|Q_l|}\\
		&=& 0.
	\end{eqnarray*}

	Suppose $M_{\vec{P}}(\vec{a}, \vec{b}) = 1$. In this case, we have $\vec{P} = \vec{P}^*$. Thus, for every $l \in [d]$, Eq.~(\ref{eq:tauP^*=tauP}) implies that $M[T](\vec{a}, \vec{b}) = 1$ if $T=Q_l$, and $M[T](\vec{a}, \vec{b}) = 0$, otherwise. Therefore,
	\begin{eqnarray}
		H[Q_l](\vec{a}, \vec{b})
		&=&  \sum_{T \in \mathcal{T}(Q_l)}  (-1)^{|T| - |Q_l|} M[T](\vec{a}, \vec{b})\\
		&=&   (-1)^{|Q_l| - |Q_l|}\\
		&=& 1.
	\end{eqnarray}
	Consequently
	\[
	H[Q_1](\vec{a}, \vec{b}) \circ H[Q_2](\vec{a}, \vec{b}) \circ \dots \circ H[Q_d](\vec{a}, \vec{b})= 1,
	\]
	as desired.

\subsection{Proof of Lemma~\ref{lem:UB}}

Let $z = (e-1)(d-1)$. For simplicity, we assume that $n$ is divisible by $z$. Let $f:\dom \to \{ 0, 1\}$ be $z$-piecewise with $f_1, f_2, \dots , f_z$. We first prove the proposition for the case where $f$ satisfies
\[
f(\vec{x}, \vec{y}) = \bigvee_{j=1}^z f_j(\vec{x}, \vec{y}).
\]

Let us relabel $f_l$ for $l \in [z]$ as $f_{k, l}$ for $k \in [e-1]$, and $l \in [d-1]$. We then denote by $B_{k, l} =\{ i\in [n] \mid \mbox{$x_i$ or $y_i$ is fed into $f_{k,l}$} \}$ for $k \in [e-1]$, and $l \in [d-1]$, each of which contains $n/z$ integers. Thus, we have
\[
f(\vec{x}, \vec{y}) = \bigvee_{k\in [e-1]} \bigvee_{l \in [d-1]} f_{k,l}(\vec{x}, \vec{y}).
\]
By the assumption, $f_{k,l}$ is computable by a selective neural set $S_{k, l}$ for every pair of $k\in [e-1]$ and $l \in [d-1]$.

We construct the desired threshold circuit $C$ by arranging and connecting the selective neural sets, where $C$ has a simple layered structure consisting of the selective neural sets. After we complete the construction of $C$, we show that $C$ computes $f$, and evaluate its size, depth, energy, and weight. 

We start from the bottom level. First, we put in the first level the gates in $S_{k,1}$ for every $k \in [e-1]$. Then, for each $l$, $2 \le l \le d-1$, we add at the $l$th level the gates $g \in S_{k,l}$ for every $k \in [e-1]$, and connect the outputs of all the gates at the lower level to every $g \in S_{k, l}$ with weight $-2nw'/z$. For $g \in S_{k,l}$ and $(\vec{a}, \vec{b}) \in \dom$, we denote by $g^C(\vec{a}, \vec{b})$ the output of the gate $g$ placed in $C$. Finally, we add $\gtop$ that computes a conjunction of all gates in the layers $1, 2 \dots , d-1$:
\[
\gtop(\vec{x}, \vec{y}) = \sign \left( \sum_{k \in [e-1]} \sum_{l \in [d-1]} \sum_{g \in S_{k, l}} g^C(\vec{x}, \vec{y}) - 1\right).
\]

We here show that $C$ computes $f$. By construction, the following claim is easy to verify:
\begin{claim}\label{clm:g^B_kl=F^B_kl}
	For any $g \in S_{k,l}$,
	\[
	g^C(\vec{x}, \vec{y}) = \left\{
	\begin{array}{ll}
		g\inp{x}{y} & \mbox{if every gate at the levels $1, \dots , l-1$ outputs zero};\\
		0 & \mbox{otherwise}.
	\end{array}
	\right.
	\]
\end{claim}
\begin{proof}
	If every gate at the levels $1, \dots , l-1$ outputs zero, the output of $g^C$ is identical to the counterpart of $g$, and hence $g^C(\vec{x}, \vec{y}) = g(\vec{x}, \vec{y})$. Otherwise, there is a gate outputting one at the lower levels. Since $g^C$ receives an output from the gate at the lower level, the value $-2w'n/z$ is subtracted from the potential of $g^C$. Since $g$ receives at most $2n/z$ positive weights bounded by $w'$, the potential of $g^C$ is now below its threshold. Otherwise, $g$ outputs one for any input assignment, which implies that, since $S_{k, l}$ is selective, $f_{k, l}$ is the constant function. This contradicts the fact that $f$ is $z$-piecewise.
\end{proof}

Suppose $f(\vec{a}, \vec{b})=0$. In this case, for every $k \in [e-1]$, $l \in [d-1]$, and $g \subseteq S_{k, l}$, $g{k,l}(\vec{a}, \vec{b}) = 0$. Therefore, Claim~\ref{clm:g^B_kl=F^B_kl} implies that no gate in $C$ outputs one.

Suppose $f(\vec{a}, \vec{b})=1$. In this case, there exists $l^* \in [d-1]$ such that $f_{k, l^*}(\vec{a}, \vec{b}) = 1$ for some $k \in [e-1]$, while $f_{k, l}(\vec{a}, \vec{b}) = 0$ for every $1 \le k \le e-1$ and $1 \le l \le l^*-1$. Since $S_{k, l^*}$ compute $f_{k, l}$, Claim~\ref{clm:g^B_kl=F^B_kl} implies that there exists $g \in S_{k, l^*}$ such that $g^C(\vec{a}, \vec{b}) = 1$, which implies that $\gtop\inp{x}{y} = 1$.

Finally, we evaluate the size, depth, energy, and weight of $C$. Since $C$ contains at most $s'$ gates for each pair of $k \in [e-1]$ and $l\in [d-1]$ where $z=(e-1)(d-1)$, we have in total $s \le zs' + 1$. The additional one corresponds to the output gate. Because the gates $g \in S_{k, l}$ are placed at the $l$-th level for $l \in [d-1]$, the level of $\gtop$ is clearly $d$, and hence, $C$ has depth $d$. Claim~\ref{clm:g^B_kl=F^B_kl} implies that if there is a gate outputting one at level $l$, then no gate in higher levels outputs one. In addition, since $S_{k, l}$ is selective, at most one gate $g \in S_{k, l}$ outputs one. Therefore, at most $e-1$ gates at the $l$th level output one, followed by $\gtop$ gates outputting one. Thus, $C$ has an energy $e$. Any connection in $C$ has weight at most $w'$ or $2nw'/z$. Thus, the weight of $C$ is $2nw'/z$.

Consider the other case where
\[
f(\vec{x}, \vec{y}) = \overline{\bigvee_{j=1}^z f_j(\vec{x}, \vec{y})}.
\]
We can obtain the desired circuit $C$ by the same construction as above except that $\gtop$ computes the complement of a conjunction of all gates in the layers $1, 2 \dots , d-1$:
\[
\gtop(\vec{x}, \vec{y}) = \sign \left( \sum_{k \in [e-1]} \sum_{l \in [d-1]} \sum_{g \in S_{k, l}} -g^C(\vec{x}, \vec{y}) \right).
\]

\subsection{Proof of Theorem~\ref{thm:UB_conj}}
For simplicity, we consider the case where $n$ is a multiple of $z$. It suffices to show that $\conj_n$ is $z$-piecewise, and computable by a neural set of size $s' = 2^{n/z}$ and weight $w' = n/z$.

We can verify that $\conj_n$ is $z$-piecewise, since
\[
\conj_n\inp{x}{y} = \bigvee_{j=1}^z f_j (\vec{x}, \vec{y})
\]
where
\[
f_j (\vec{x}, \vec{y}) = \bigvee_{i\in B_j} x_i \wedge y_i.
\]
and $B_j = \{ i \in [n] \mid (j-1)\lceil n/z \rceil +1 \le i \le  jn/z \}$.

Consider an arbitrary fixed $j \in [z]$. Below we show that $f_j$ is computable by a neural set of  size
\[
s' \le 2^{n/z}
\]
and weight $w' \le n/z$.

Recall that we denote by $\true{\mathrm{P}}$ for a statement $\mathrm{P}$  a function that outputs one if $\mathrm{P}$ is true, and zero otherwise. Then $f_j$ can be expressed as
\begin{eqnarray}\label{eq:f_j=F_j^B}
	f_j (\vec{x}, \vec{y}) = \bigvee_{\emptyset \subset B \subseteq B_j} F^B_j(\vec{x}, \vec{y})
\end{eqnarray}
where
\[
F^B_j(\vec{x}, \vec{y}) = \bigwedge_{i \in B} \true{x_i = 1} \wedge \bigwedge_{i \not\in B} \true{x_i = 0}  \wedge \bigvee_{i \in B} y_i.
\]
For any assignment $\vec{a} \in \{ 0, 1\}^n$, we define $B^*(\vec{a})$ as
\[
B^*(\vec{a})= \{ i \in B_j \mid a_i = 1 \}.
\]
Then, for every $\inp{a}{b} \in \dom$,
\begin{eqnarray}\label{eq:F^B_j=1_0}
	F^B_j(\vec{a}, \vec{b}) = \left\{
	\begin{array}{ll}
		1 & \mbox{ if } B = B^*(\vec{a}) \mbox{ and } \exists i \in B, b_i = 1; \\
		0 & \mbox{ otherwise}.
	\end{array}
	\right.
\end{eqnarray}
The function $F^B_j$ is computable by a threshold gate.
\begin{claim}\label{clm:g_computes_F^B}
	For any $B \subseteq B_j$, $F^B_j$ can be computed by a threshold gate $g^B_j$ with weights $w^x_i$, $w^y_i$; the threshold $t$ is defined as follows: for every $i \in [n]$,
	\[
	w^x_i = \left\{
	\begin{array}{ll}
		0 & \mbox{ if } i \not\in B_j;\\
		|B| & \mbox{ if } i \in B \subseteq B_j;\\
		-|B| & \mbox{ if } i \in B_j \backslash B,
	\end{array}
	\right.
	\]
	and
	\[
	w^y_i = \left\{
	\begin{array}{ll}
		0 & \mbox{ if } i \not\in B_{k, l};\\
		1 & \mbox{ if } i \in B \subseteq B_{k, l};\\
		0 & \mbox{ if } i \in B_{k, l} \backslash B,
	\end{array}
	\right.
	\]
	and $t = |B|^2 + 1$.
\end{claim}
\begin{proof}
	Suppose $F^B_j\inp{x}{y} = 1$, that is, $B = B^*(\vec{x})$, and there exists $i^* \in B$ such that $y_{i^*} = 1$. Then, we have
	\[
	\sum_{i \in B_j} w^x_i x_i + \sum_{i \in B} w^y_i y_i \ge |B|\cdot |B| + 1 = t;
	\]
	thus, $g^B_j$ outputs one.
	
	Suppose $F^B_j\inp{x}{y} = 0$. There are two cases: $B \neq B^*(\vec{x})$ and $y_i = 0$ for every $i \in B_j$
	
	Consider the first case. If there exists $i^* \in B^*(\vec{x}) \backslash B$, then
	\[
	\sum_{i \in B_j} w^x_i x_i \le \sum_{i \in B} w^x_i x_i + w^x_{i^*} \le |B|\cdot |B| - |B|.
	\]
	Thus,
	\[
	\sum_{i \in B_j} w^x_i x_i + \sum_{i \in B} w^y_i y_i \le (|B|\cdot |B| - |B|) +|B| \le |B|^2 < t,
	\]
	and consequently, $g^B_{k, l}$ outputs zero. If $B^*(\vec{x})  \subset B$, then
	\[
	\sum_{i \in B_j} w^x_i x_i \le (|B|-1)\cdot |B|.
	\]
	Thus,	
	\[
	\sum_{i \in B_j} w^x_i x_i + \sum_{i \in B} w^y_i y_i \le ((|B|-1)\cdot |B|) + |B| \le |B|^2 < t,
	\]
	and hence, $g^B_j$ outputs zero.
	
	In the second case, we have
	\[
	\sum_{i \in B} w^y_i y_i = 0.
	\]
	Thus,
	\[
	\sum_{i \in B_j} w^x_i x_i + \sum_{i \in B} w^y_i y_i \le |B|\cdot |B| < t,
	\]
	and hence, $g^B_j$ outputs zero.	
\end{proof}

For any $\vec{a} \in \{ 0, 1\}^n$, Eq.~(\ref{eq:F^B_j=1_0}) implies that only $g^{B^*(\vec{a})}_j$ are allowed to output one. Thus, by Eq.~(\ref{eq:f_j=F_j^B}), a selective neural set
\[
S_j =\{g^B_j \mid \emptyset \subset B  \subseteq B_j\} 
\]
computes $f_j$. Since $|B_j| \le n/z$, we have $|S_j| \le 2^{n/z}$. Claim~\ref{clm:g_computes_F^B} implies that $w' \le n/z$.

\subsection{Proof of Theorem~\ref{thm:UB_eq}}
Let $z =(e-1)(d-1)$. $\eq_n$ is $z$-piecewise, since
\[
\eq_n\inp{x}{y} = \overline{\bigvee_{j=1}^z f_j (\vec{x}, \vec{y})}
\]
where
\[
f_j (\vec{x}, \vec{y}) = \bigvee_{i\in B_j} x_i \oplus y_i
\]
and $B_j = \{ i \in [n] \mid (j-1)\lceil n/z \rceil +1 \le i \le  jn/z \}$. Then the theorem implies that $f_j$ is computable by a selective neural set of size $s'=2^{2n/z}$ and $w' = 1$.

\subsection{Proof of Theorem~\ref{thm:LB_discretizer}}
Let $\dis$ be a discretizer and $\af$ be an activation function such that $\dis \circ \af$ has a silent range for $I$. In the proof, we only consider an open interval $I = (t_{\min}, t_{\max})$ because the proofs for the other cases are similar. Let $C$ be a $(\dis \circ \af)$-circuit of size $s$, depth $d$, energy $e$, and weight $w$. We obtain the desired threshold circuit $C'$ by showing a procedure by which any $(\dis \circ \af)$-gate $g$ in $C$ can be safely replaced by a set of threshold gates.

Let $g$ be an arbitrary $(\dis \circ \af)$-gate in $C$ that computes $g(\vec{x}, \vec{y}) = \dis \circ \af \left( p\inp{x}{y}. \right)$ We first consider $[t_{\max}, \infty)$. Let $P_g$ be a set of potential values for which $g$ outputs a non-zero value:
\[
P_g = \{ p\inp{a}{b} \mid \inp{a}{b}\in \dom, t_{\max} \le p\inp{a}{b} \}.
\]
Since the activation function and weights are discretized, we have $|P_g| = O((s+n)w)$.	

We operate the binary search over $P_g$, and construct a threshold gate that outputs one for every input $\inp{a}{b}$ such that $p\inp{a}{b}$ takes a particular value in $P$. For any $Q \subseteq P_g$, we define $\mathrm{mid}(Q)$ as the median of the integers in $Q$, and $Q^+$ (resp., $Q^-$) as the upper (resp., lower) half of $Q$:
\[
Q^+ = \{ p \in Q \mid p \le \mathrm{mid}(Q)\} \quad \mbox{ and } \quad 	Q^- = \{ p \in Q \mid \mathrm{mid}(Q) < p\}.
\]
If $Q$ contains an even number of values, we take the greater value of the two median values .

Let $\vec{s}$ be a binary string. We inductively construct a threshold gate $g_{\vec{s}}$ on the length of $\vec{s}$. For two strings $\vec{t}$ and $\vec{s}$, we write $\vec{t} \prec \vec{s}$ if $\vec{t}$ is a proper prefix of $\vec{s}$. We denote  a string $\vec{t}$ followed by 0 (resp., by 1) by $\vec{t}0$ (resp., $\vec{t}1$).

As the base of our construction, we consider the empty string $\epsilon$. Let $P_\epsilon = P_g$. We constructs a threshold gate $g_\epsilon$ that computes
\[
g_\epsilon(\vec{x}, \vec{y}) = \sign \left( p \inp{x}{y} - t_g[\epsilon] \right),
\]
where $t_g[\epsilon] = \mathrm{mid}(P_\epsilon)$.

Suppose we have constructed gates $g_{\vec{t}}$ for every $\vec{t}$ satisfying $|\vec{t}| \le k-1$. Consider a string $\vec{s}$ of length $k$. By the induction hypothesis, we have gates $g_{\vec{t}}$ for every $\vec{t}$ and $\vec{t} \prec \vec{s}$. Let $\vec{s}'$ be a string obtained by dropping the last symbol of $\vec{s}$. Then, we define
\[
P_\vec{s} = \left\{
\begin{array}{ll}
	P_{\vec{s}'}^+ & \mbox{if the last symbol of $\vec{s}$ is 1};\\
	P_{\vec{s}'}^- & \mbox{if the last symbol of $\vec{s}$ is 0}.
\end{array}
\right.
\]
Let $W = 3(s +n)w$. We  construct a threshold gate $g_{\vec{s}}$ as follows:
\[
g_{\vec{s}}(\vec{x}, \vec{y}) = 
\sign \left( p \inp{x}{y} + \sum_{\vec{t} \prec \vec{s}} w_{\vec{t}, \vec{s}} \cdot  g_{\vec{t}}(\vec{x}, \vec{y}) - t_g[{\vec{s}}] \right),
\]
where $w_{\vec{t}, \vec{s}}$ is the weight of the output of $g_{\vec{t}}$ and is defined as
\[
w_{\vec{t}, \vec{s}} = \left\{
\begin{array}{ll}
	W & \mbox{if $\vec{t}1$ is a prefix of $\vec{s}$};\\
	-W & \mbox{if $\vec{t}0$ is a prefix of $\vec{s}$};\\
	0 & \mbox{otherwise},
\end{array}
\right.
\]
and $t_g[\vec{s}] = \mathrm{mid}(P_\vec{s})  - W\cdot N_1(\vec{s})$, where $N_1(\vec{s})$ is the number of ones in $\vec{s}$.

We repeatedly apply the above procedure until $|P_\vec{s}| = 1$. Since we apply the binary search over $P$, we have $O(|P|) = O((s+n)w)$ gates and the length of $|\vec{s}|$ as $O(\log (s+n) + \log w)$ for any gate $g_{\vec{s}}$. 

Consider the strings $\vec{s}$ for which we have constructed $g_{\vec{s}}$. Let \[
S_g = \{ \vec{s} \mid 2 \le |P_\vec{s}| \} \quad \mbox{ and } \quad L_g = \{ \vec{s} \mid |P_\vec{s}| = 1 \}.
\]
For each $\inp{a}{b} \in \dom$, we denote the unique string that satisfies $P_{\vec{s}^*\inp{a}{b}} = \{ p\inp{a}{b }\}$ by $\vec{s}^* \inp{a}{b} \in L_g$. The following claims show that the $g_\vec{s}$s are useful for simulating $g$. 		
\begin{claim}\label{clm:t_max<p}
	Let $\inp{a}{b} \in \dom$ be an input assignment that satisfies $t_{\max} \le p\inp{a}{b}$.
	\begin{description}
		\item[(i)] For $\vec{s} \in S_g$, $g_\vec{s}$ outputs one if and only if $\vec{s}1$ is a prefix of $\vec{s}^* \inp{a}{b}$.
		\item[(ii)] For $\vec{s} \in L_g$, $g_{\vec{s}}$ outputs one if and only if $\vec{s} = \vec{s}^*\inp{a}{b}$.
	\end{description}  
\end{claim}
\begin{proof}
	Consider an arbitrary input assignment  $\inp{a}{b} \in \dom$ that satisfies $t_{\max}  \le p\inp{a}{b}$. For notational simplicity, we write $\vec{s}^*$ for $\vec{s}^* \inp{a}{b}$.
	
	\smallskip
	\noindent
	Proof of (i). We verify the claim by induction on the length of $\vec{s}$. For the base case, we consider $\vec{\epsilon}$. It suffices to show that $g_\epsilon \inp{a}{b} = 1$ if the first symbol of $\vec{s}^*$ is 1; otherwise, $g_\epsilon \inp{a}{b} = 0$. If the first symbol is 1, we have $p \inp{a}{b} \in P^+_\epsilon$, which implies that $t_g[\epsilon] \le p \inp{a}{b}$. Thus, $g_\epsilon \inp{a}{b} = 1$. Similarly, if the first symbol is zero, we have $p \inp{a}{b} \in P^-_\epsilon$, implying that $p \inp{a}{b} < t_g[\epsilon] $. Thus, $g_\epsilon \inp{a}{b} = 0$.
	
	We assume that for the induction hypothesis, $g_\vec{t}$ outputs one if and only if $\vec{t}1$ is a prefix of $\vec{s}^* \inp{a}{b}$ for every $\vec{t}$ of length $k-1$, at most, for a positive integer $k$. Next, we consider a string $\vec{s}$ of length $k$.
	
	We first verify that $g_\vec{s}$ outputs zero if $\vec{s}$ itself is not a prefix of $\vec{s}^* \inp{a}{b}$.
	If $\vec{s}$ is not a prefix of $\vec{s}^*$, there exists a prefix $\vec{t}'$ of $\vec{s}$ such that $\vec{t}'0$ (resp., $\vec{t}'1$) is a prefix of $\vec{s}$, whereas $\vec{t}'1$ (resp., $\vec{t}'0$) is not a prefix of $\vec{s}^*$.
	
	Consider the case where $\vec{t}'0$ is a prefix of $\vec{s}^*$ (i.e., $\vec{t}'1$ is a prefix of $\vec{s}$). In this case, the induction hypothesis implies that: 
	$g_{\vec{t}'} \inp{a}{b} = 0$. In addition, since $\vec{t}'1$ is a prefix of $\vec{s}$, we have $w_{\vec{t}', \vec{s}} = W$. Thus, the potential of $g_{\vec{s}}$ for $\inp{a}{b}$ is at most
	\begin{eqnarray*}
		&& p \inp{a}{b} + \sum_{\vec{t} \prec \vec{s}} w_{\vec{t}, \vec{s}} \cdot  g_{\vec{t}}(\vec{a}, \vec{b}) - t_g[{\vec{s}}]\\
		&\le& p \inp{a}{b} +W \cdot (N_1(\vec{s})-1) - \mathrm{mid}(P_{\vec{s}}) - W\cdot N_1(\vec{s})\\
		&\le& p \inp{a}{b} - \mathrm{mid}(P_{\vec{s}}) - W, 
	\end{eqnarray*} 
	which is less than zero because $p \inp{a}{b} \le (s+n)w$.
	
	Consider the case in which $\vec{t}'1$ is a prefix of $\vec{s}^*$ (i.e., $\vec{t}'0$ is a prefix of $\vec{s}$). In this case, the induction hypothesis implies that $g_{\vec{t}'} \inp{a}{b} = 1$. In addition, since $\vec{t}'0$ is a prefix of $\vec{s}$, $w_{\vec{t}', \vec{s}} = -W$. Thus, the potential of $g_{\vec{s}}$ for $\inp{a}{b}$ is at most
	\begin{eqnarray*}
		&&p \inp{a}{b} + \sum_{\vec{t} \prec \vec{s}} w_{\vec{t}, \vec{s}} \cdot  g_{\vec{t}}(\vec{a}, \vec{b}) - t_g[{\vec{s}}]\\
		&\le& p \inp{a}{b}  + (W \cdot N_1(\vec{s}) - W) - \mathrm{mid}(P_{\vec{s}}) - W\cdot N_1(\vec{s})\\
		&\le& p \inp{a}{b} - \mathrm{mid}(P_{\vec{s}}) - W, 
	\end{eqnarray*} 
	which is again less than zero because $p \inp{a}{b} \le 2(s+n)w$.
	
	Suppose $\vec{s}$ is a prefix of $\vec{s}^*$. The induction hypothesis implies that the potential of $g_\vec{s}$ is equal to
	\begin{eqnarray*}
		&&p \inp{a}{b} + \sum_{\vec{t} \prec \vec{s}} w_{\vec{t}, \vec{s}} \cdot  g_{\vec{t}}(\vec{a}, \vec{b}) - t_g[{\vec{s}}]\\
		&=& p \inp{a}{b}  + W \cdot N_1(\vec{s})  - \mathrm{mid}(P_{\vec{s}}) - W\cdot N_1(\vec{s})\\
		&=& p \inp{a}{b} - \mathrm{mid}(P_{\vec{s}}).
	\end{eqnarray*} 
	Thus, similar to the base case, we can show that $g_\vec{s} \inp{a}{b} = 1$ if $\vec{s}1$ is a prefix of $\vec{s}^*$, and $g_\vec{s}(\inp{a}{b}) = 0$ otherwise. More formally, if $\vec{s}1$ is a prefix of $\vec{s}^*$, we have $p \inp{a}{b} \in P^+_\vec{s}$, which implies that $\mathrm{mid}(P_\vec{s}) \le p \inp{a}{b}$. Thus, $g_\vec{s} \inp{a}{b} = 1$. If $\vec{s}0$ is a prefix of $\vec{s}^*$, we have $p \inp{a}{b} \in P^-_\vec{s}$, which implies that $p \inp{a}{b} <\mathrm{mid}(P_\vec{s}) $. Thus, $g_\vec{s} \inp{a}{b} = 0$.
	
	\smallskip
	\noindent
	Proof of (ii). Consider $\vec{s} \in L_g$. Similar to claim (i), we can verify that $g_\vec{s}\inp{a}{b} = 0$ if $\vec{s} \neq \vec{s}^*$. If $\vec{s} = \vec{s}^*$, claim (i) also implies that the potential of $g_{\vec{s}^*}$ is equal to
	\begin{eqnarray*}
		&&p \inp{a}{b} + \sum_{\vec{t} \prec \vec{s}^*} w_{\vec{t}, \vec{s}^*} \cdot  g_{\vec{t}}(\vec{a}, \vec{b}) - t_g[{\vec{s}^*}]\\
		&=& p \inp{a}{b}  + W \cdot N_1(\vec{s})  - \mathrm{mid}(P_{\vec{s}}) - W\cdot N_1(\vec{s})\\
		&=& p \inp{a}{b} - \mathrm{mid}(P_{\vec{s}^*}).
	\end{eqnarray*} 
	Since $P_{\vec{s}^*} = \{p \inp{a}{b} \}$, $\mathrm{mid}(P_{\vec{s}^*}) = p\inp{a}{b}$. Thus, we have
	\begin{eqnarray*}
		p \inp{a}{b} - \mathrm{mid}(P_{\vec{s}^*}) = 0,
	\end{eqnarray*}
	which implies that $g_{\vec{s}^*}\inp{a}{b} = 1$.
\end{proof}

\begin{claim}\label{clm:p<t_max}
	For any $\inp{a}{b} \in \dom$ that satisfies $p \inp{a}{b} < t_{\max}$, every gate $g_\vec{s}$ outputs zero.
\end{claim}
\begin{proof}
	Since $p \inp{a}{b} < t_{\max} \le \mathrm{mid}({P_g})$, $g_\epsilon \inp{a}{b} = 0$, which is implied by a similar argument to Claim~\ref{clm:t_max<p}, all the gates $g_\vec{s}$ such that $\vec{s}$ contains a symbol 1. If $\vec{s}$ consists of only 0s, we have $p\inp{a}{b} < t_{\max} \le \mathrm{mid}(Q)$ for any $Q \subseteq P_g$, which implies that $g_\vec{s} \inp{a}{b} = 0$. 
\end{proof}

Claims~\ref{clm:t_max<p} and~\ref{clm:p<t_max} imply that we can safely replace $g$ with $g_\vec{s}$s, $\vec{s} \in S_g \cup L_g$, by connecting the output of each gate $g_{\vec{s}}$, $\vec{s} \in L_g$, with weight $w_{g, g'} \cdot (\delta \circ \af(p))$, where $p$ is the unique value in $P_\vec{s}$ to every gate $g'$ that the gate $g$ is originally connected to in $C$. Since $|S_g \cup L_g| = O((s+n)w)$ and the length of $\vec{s}$ is $O(\log (s+n) + \log w)$, the size and depth  increase by these factors, respectively.

We can construct another set of threshold gates for $(-\infty, t_{\min}]$ in a similar manner to the above procedure That is, it suffices to consider a gate obtained by multiplying $-1$ by the weights and threshold of $g$ together with an interval $[-t_{\min}, \infty)$.

Using the above procedure, we replace every gate $g \in G\backslash \{\gtop\}$ in $C$ with a set of threshold gates, and complete the construction of $C'$. Clearly, the size of $C'$ is $O(s\cdot (s+n)w)$. Since $g_\vec{s}$ receives the outputs of $g_{\vec{t}}$ for every $\vec{t} \prec \vec{s}$ and the length of $\vec{s}$ is $O(\log (s+n) + \log w)$, the depth of $C'$ is $O(d \cdot (\log (s+n) + \log w))$. Claims~\ref{clm:t_max<p} and~\ref{clm:p<t_max} imply that the energy of $C'$ is $O(e \cdot (\log(s+n) + \log w))$. Clearly, the weight of $C'$ is $W = O((s+n)w)$. Thus, Theorem~\ref{thm:LB_size} implies that $rk(M_C)$ is bounded by
\begin{eqnarray*}
	&& O(e(\log(s+n) + \log w) \! \cdot \! d(\log (s+n) + \log w) \! \cdot \! (\log s + \log w + \log n)) \\
	&=& O(ed(\log s + \log n + \log w)^3),
\end{eqnarray*}
as desired.

\end{document}